\documentclass[12pt]{article}
\usepackage{}
\usepackage{amsthm,amsmath,indentfirst,amssymb}
\usepackage{color}
\usepackage{mathrsfs}
\usepackage{hyperref}
\usepackage{url,CJK}
\usepackage{pstricks}
\usepackage{tikz}
\usepackage{caption}
\usepackage{subfigure}
\usepackage{algorithm}
\usepackage{algorithmic}
\usepackage{graphicx}
\usepackage[numbers,sort&compress]{natbib}

\newtheorem{lemma}{Lemma}[section]
\newtheorem{theorem}[lemma]{Theorem}

\newtheorem{definition}[lemma]{Definition}

\newtheorem{remark}[lemma]{Remark}
\newtheorem{example}[lemma]{Example}

\begin{document}\baselineskip 0.56cm


\title{A parallel algorithm for minimum weight set cover with small neighborhood property}

\author{\footnotesize Yingli Ran$^1$\quad Yaoyao Zhang$^2$\quad Zhao Zhang$^1$\thanks{Corresponding author: Zhao Zhang, hxhzz@sina.com}\\
	{\it\small $^1$ College of Mathematics and Computer Science, Zhejiang Normal University}\\
    {\it\small Jinhua, Zhejiang, 321004, China}\\
    {\it\small $^2$ College of Mathematics and System Sciences, Xinjiang University}\\
    {\it\small Urumqi, Xinjiang, 830046, China}}
\date{}
\maketitle
\begin{abstract}
{\rm This paper studies the minimum weight set cover (MinWSC) problem with a {\em small neighborhood cover} (SNC) property proposed by Agarwal {\it et al.} in \cite{Agarwal.}. A parallel algorithm for MinWSC with $\tau$-SNC property is presented, obtaining approximation ratio $\tau(1+3\varepsilon)$ in $O(L\log_{1+\varepsilon}\frac{n^3}{\varepsilon^2}+ 4\tau^{3}2^\tau L^2\log n)$ rounds, where $0< \varepsilon <\frac{1}{2}$ is a constant, $n$ is the number of elements, and $L$ is a parameter related to SNC property. Our results not only improve the approximation ratio obtained in \cite{Agarwal.}, but also answer two questions proposed in \cite{Agarwal.}.}




\end{abstract}

\section{Introduction}\label{sec1}

In this paper, we study parallel approximation algorithm for the {\em minimum weight set cover} (MinWSC) problem with a {\em small neighborhood cover} (SNC) property.

MinWSC is a classic combinatorial optimization problem. Tight approximation ratios have long been known, including $(\ln\Delta+1)$-approximation \cite{Chvatal} and $f$-approximation \cite{Bar-Yehuda and Even}, where $\Delta$ is the size of a maximum set, $f$ is the maximum frequency of elements (that is, the number of sets containing a common element).

With the fast development in computer architecture and the increasing number of CPU cores, designing efficient parallel algorithms has emerged as an active research area in recent years. There are many parallel algorithms for MinWSC \cite{Rajagopalan and Vazirani,Khuller.,Blelloch,Berger}. In particular, Khuller {\it et al.} \cite{Khuller.} gave a parallel algorithm with approximation ratio $f/(1-\varepsilon)$ where $0<\varepsilon< 1$ is a constant. Note that $f$ might be $\Omega(n)$ in a worst general case. For the problem with special structural properties, it might be possible to get a better approximation ratio. For example, in \cite{Agarwal.}, Agarwal {\it et al.} proposed a structural property called {\em small neighborhood cover} (SNC). A lot of important MinWSC problems possess SNC property, such as the {\em vertex cover problem} (VC), {\em interval cover problem} (IC), {\em tree cover problem} (TC), {\em interval hitting problem}, {\em priority interval cover problem}, {\em bag interval cover problem}, etc. A parallel algorithm for MinWSC with $\tau$-SNC property was first studied in \cite{Agarwal.,Agarwal1.}, where \cite{Agarwal1.} is the preliminary version of \cite{Agarwal.}. \cite{Agarwal.} obtained approximation ratio at most $1+(2+\varepsilon)\tau^2$ in $O(L^2+(1/\varepsilon)L\log m)$ rounds, where $m$ is the number of sets, and $L$ is the depth of the $\tau$-SNC decomposition. Note that in the above mentioned problems, $\tau$ is much smaller than $f$.

In the conclusion part of \cite{Agarwal.}, three questions were proposed. In this paper, we give positive answers to two of them, and present an improved parallel algorithm.

\subsection{Related work}\label{sec01}

For MinWSC, Chv$\acute{a}$tal \cite{Chvatal} gave a greedy algorithm achieving approximation ratio $H(\Delta)$, where $\Delta$ is the size of a maximum set and $H(\Delta)=\sum_{i=1}^{\Delta}\frac{1}{i}$ is the $\Delta$-th Harmonic number (note that $\ln\Delta\leq H(\Delta)\leq\ln\Delta+1$). This ratio is tight under the assumption $P\neq NP$ \cite{Dinur, Feige}. In \cite{Bar-Yehuda and Even}, Bar-Yehuda and Even used prime-dual schema to obtain an approximation ratio $f$, where $f$ is the maximum number of sets containing a common element. This ratio is tight under the Unique Games Conjecture \cite{Kho}. Note that these are all sequential algorithms.

Considering parallel algorithms for MinWSC, Berger {\it et al.} \cite{Berger} gave a parallel algorithm with approximation ratio $(1+\varepsilon)H(\Delta)$ in $O(\log^2n\log m\log^2(nm)/\varepsilon^6)$ rounds, where $n$ is the number of elements and $m$ is the number of sets. In \cite{Rajagopalan and Vazirani}, using a primal-dual schema, Rajagopalan and Vazirani gave a parallel algorithm with improved number of rounds $O(\log n\log (nm)\log(nm\Delta))$ and a weaker approximation ratio $2(1+\varepsilon)H(\Delta)$. In \cite{Blelloch}, by proposing a concept called nearly independent set, Blelloch {\it et al.} were able to further improve the approximation ratio to $(1+\varepsilon)H(\Delta)$ in $O(\log^3 M)$ rounds, where $M$ is the sum of sizes of all sets in $\mathcal{S}$. These are parallel algorithms for MinWSC achieving approximation ratio measured in terms of $\Delta$. In \cite{Khuller.}, by a primal-dual method, Khuller {\it et al.} presented a parallel algorithm for MinWSC with approximation ratio $f/(1-\varepsilon)$ in $O(f\log^2 m\log(\frac{1}{\varepsilon}))$ rounds. For the partial version of MinWSC, the goal of which is to cover not all elements, but at least some percentage of the elements, Ran {\it et al.} \cite{Ran} presented a parallel algorithm with approximation ratio at most $\frac{f}{1-2\varepsilon}$ in $O(\frac{1}{\varepsilon}\log \frac{mn}{\varepsilon})$ rounds, where $0<\varepsilon<\frac{1}{2}$ is a constant.

In \cite{Agarwal.}, Agarwal {\it et al.} proposed a property called $\tau$-SNC, which applies to many problems including VC, TC, IC and some other graph structural covering problems. Using primal-dual schema, they presented a parallel algorithm with approximation ratio $1+(2+\varepsilon)\tau^2$ in $O(L^2\log n+(1/\varepsilon)L\log n)$ rounds, where $L$ is the depth of $\tau$-SNC decomposition. They also gave a distributed algorithm for MinWSC with $\tau$-SNC property and obtained approximation ratio $\tau$ in $O(L\log n+\log^2 n)$ communication rounds.

The parallel algorithm in \cite{Agarwal.} consists of a forward phase and a deletion phase. In the conclusion part, three questions were proposed:

$Q_{\uppercase\expandafter{\romannumeral1}}\colon$ In the forward phase, can one construct ($1-\varepsilon$)-maximal solutions via a procedure having running time independent of $f$?

$Q_{\uppercase\expandafter{\romannumeral2}}\colon$ In the deletion phase, the algorithm produces a solution satisfying the primal slackness property with parameter $\tau^2$. Can this be improved to $\tau$?

$Q_{\uppercase\expandafter{\romannumeral3}}\colon$ The deletion phase leads to $L^2$ iterations. Can this be improved to $L$?

\subsection{Our Contributions}

In this paper, using an improved primal-dual method, we give a parallel algorithm for MinWSC with $\tau$-SNC property, achieving approximation ratio $\tau(1+3\varepsilon)$ in $O(L\log_{1+\varepsilon}\frac{n^3}{\varepsilon^2}+L^2\log n)$ rounds where $0< \varepsilon <\frac{1}{2}$ is a constant. This work not only improves the approximation ratio in \cite{Agarwal.}, but also answers two of the three questions in \cite{Agarwal.} positively.

In the forward phase of \cite{Agarwal.}, the authors used a parallel primal-dual algorithm to obtain a $\lambda$-maximal solution with $\lambda<1/2$. In our algorithm, we use a different idea to increase dual variables in a geometric series and obtain a $(1-\varepsilon)$-maximal solution with the same number of rounds as that of \cite{Agarwal.}. This gives a positive answer to $Q_{\uppercase\expandafter{\romannumeral1}}$.

In the deletion phase of \cite{Agarwal.}, the authors dealt the elements from layer $L$ down to layer 1. In layer $k$, they tried to find maximal independent sets to obtain a solution satisfying the primal slackness property with parameter $\tau^2$ in  $O((L-k+1)\log n)$ rounds. In this paper, we managed to decrease the primal slackness parameter $\tau^2$ down to $\tau$. To realize such a goal, we proposed a random selection method such that the number of rounds to reduce the parameter from $\tau^2$ down to $\tau$ is bounded by $4\tau^{3}2^\tau \log n$. A crucial trick is: how to guarantee that with a constant probability, the number of sets covering a bad element (that is, an element which is covered more than $\tau$ times by the current collection of sets) can be strictly reduced, while the feasibility (that is, all elements are still covered) is maintained. This leads to a positive answer to $Q_{\uppercase\expandafter{\romannumeral2}}$.

The remaining part of this paper is organized as follows. In Section \ref{sec2}, we introduce some terminologies and definitions used in this paper. In Section \ref{sec3}, we present our parallel algorithm for MinWSC with SNC property and give strict analysis. Section \ref{sec5} concludes the paper and gives further discussions.

\section{Preliminaries}\label{sec2}

In this section, we give some  terminologies and definitions used in this paper.

\begin{definition}[MinWSC]
{\rm Given a weighted set system $\mathcal G=(X,\mathcal S,w)$, where $X$ is a ground set, $\mathcal S\subseteq 2^X$ is a subcollection of subsets of $X$, and $w$ is a weight function on $\mathcal S$, the goal of MinWSC is to find a minimum weight subcollection of $\mathcal S$ to cover all elements, where the set of elements covered by subcollection $\mathcal S'\subseteq \mathcal S$, denoted as $\mathcal U(\mathcal S')$, is $\bigcup_{S\in \mathcal S'}S$, and the weight $w(\mathcal S')=\sum_{S\in \mathcal S'}w(S)$.}
\end{definition}
For any $\mathcal S'\subseteq \mathcal S$, let $\mathcal{N}_{\mathcal S'}(e)=\{S\in \mathcal S'\colon e\in S\}$ be the subcollection of $\mathcal S'$ consisting of sets containing $e$. Denote $\mathcal{R}_{\mathcal S'}(e,e^\prime)=\mathcal{N}_{\mathcal S'}(e)\cap \mathcal{N}_{\mathcal S'}(e')$ the collection of sets containing both $e$ and $e'$. We say that
\begin{equation}\label{eq1220}
\mbox{$e$ and $e'$ are {\em neighbors} in $\mathcal S'$ if $e$ and $e'$ are both in some $S\in \mathcal S'$.}
\end{equation}
In other words, neighbors of $e$ in $\mathcal S'$ constitute the set $\mathcal U(\mathcal N_{\mathcal S'}(e))$. Note that $e$ is a neighbor of itself by this definition.

For an easier understanding of SNC property, let us first consider the interval cover problem, in which a set of points on a line is to be covered by the minimum number of intervals chosen from a given collection of intervals. Note that an interval cover instance can be viewed as a set cover instance by viewing each interval as an element set containing those points contained in this interval (see Fig. \ref{fig0112-1} for an illustration). Observe that in any {\em minimal} interval cover $\mathcal I'$, any point belongs to at most two intervals of $\mathcal I'$. In fact, for any point $p\in X$, let $I_\ell$ and $I_r$ be the intervals in $\mathcal N_{\mathcal I'}(p)$ with the leftmost and the rightmost endpoints, respectively, then $I_\ell\cup I_r$ can cover all points in $\mathcal U(\mathcal N_{\mathcal I'}(p))$. For example, in Fig. \ref{fig0112-1}, point $p_7$ belongs to four intervals $I_3,I_{10},I_{12},I_{16}$. If $I_{10}$ and $I_{12}$ are taken, then all neighbors of $p_7$ (namely $\{p_3,p_4,\ldots,p_{10}\}$) are covered by $I_{10}\cup I_{12}$, and thus $I_3,I_{16}$ are not needed. Such a property is preferred because reducing the frequency of points in the selected subcollection of sets will lead to better approximation ratio.

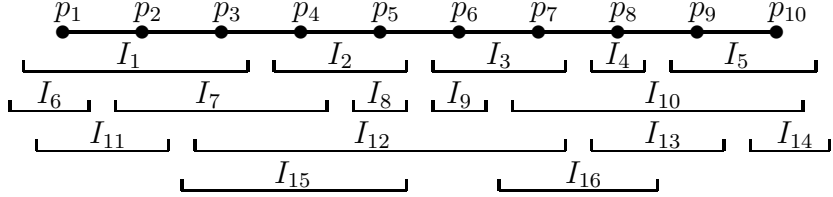
\begin{figure}[htpb]
\begin{center}
\begin{picture}(320,80)
\put(17,66){$p_1$}\put(47,66){$p_2$}\put(77,66){$p_3$}\put(107,66){$p_4$}\put(137,66){$p_5$}
\put(167,66){$p_6$}\put(197,66){$p_7$}\put(227,66){$p_8$}\put(257,66){$p_9$}\put(287,66){$p_{10}$}
\put(20,60){\linethickness{0.3mm}\line(1,0){270}}\multiput(20,60)(30,0){10}{\circle*{5}}

\put(40,48){$I_1$}\put(5,45){\line(1,0){85}}\multiput(5,45)(85,0){2}{\linethickness{0.3mm}\line(0,1){4}}
\put(120,48){$I_2$}\put(100,45){\line(1,0){50}}\multiput(100,45)(50,0){2}{\linethickness{0.3mm}\line(0,1){4}}
\put(180,48){$I_3$}\put(160,45){\line(1,0){50}}\multiput(160,45)(50,0){2}{\linethickness{0.3mm}\line(0,1){4}}
\put(225,48){$I_4$}\put(220,45){\line(1,0){20}}\multiput(220,45)(20,0){2}{\linethickness{0.3mm}\line(0,1){4}}
\put(270,48){$I_5$}\put(250,45){\line(1,0){55}}\multiput(250,45)(55,0){2}{\linethickness{0.3mm}\line(0,1){4}}

\put(10,33){$I_6$}\put(0,30){\line(1,0){30}}\multiput(0,30)(30,0){2}{\linethickness{0.3mm}\line(0,1){4}}
\put(70,33){$I_7$}\put(40,30){\line(1,0){80}}\multiput(40,30)(80,0){2}{\linethickness{0.3mm}\line(0,1){4}}
\put(135,33){$I_8$}\put(130,30){\line(1,0){20}}\multiput(130,30)(20,0){2}{\linethickness{0.3mm}\line(0,1){4}}
\put(165,33){$I_9$}\put(160,30){\line(1,0){20}}\multiput(160,30)(20,0){2}{\linethickness{0.3mm}\line(0,1){4}}
\put(240,33){$I_{10}$}\put(190,30){\line(1,0){110}}\multiput(190,30)(110,0){2}{\linethickness{0.3mm}\line(0,1){4}}

\put(30,18){$I_{11}$}\put(10,15){\line(1,0){50}}\multiput(10,15)(50,0){2}{\linethickness{0.3mm}\line(0,1){4}}
\put(130,18){$I_{12}$}\put(70,15){\line(1,0){140}}\multiput(70,15)(140,0){2}{\linethickness{0.3mm}\line(0,1){4}}
\put(240,18){$I_{13}$}\put(220,15){\line(1,0){50}}\multiput(220,15)(50,0){2}{\linethickness{0.3mm}\line(0,1){4}}
\put(290,18){$I_{14}$}\put(280,15){\line(1,0){30}}\multiput(280,15)(30,0){2}{\linethickness{0.3mm}\line(0,1){4}}

\put(100,3){$I_{15}$}\put(65,0){\line(1,0){85}}\multiput(65,0)(85,0){2}{\linethickness{0.3mm}\line(0,1){4}}
\put(210,3){$I_{16}$}\put(185,0){\line(1,0){60}}\multiput(185,0)(60,0){2}{\linethickness{0.3mm}\line(0,1){4}}
\end{picture}
\caption{An example of interval cover.}\label{fig0112-1}
\end{center}
\end{figure}

Even better in the above example, the leftmost point belongs to only one interval in any minimal solution: among all selected intervals covering this point, the one with the rightmost endpoint covers all its neighbors. Call those points satisfying such a better property as {\em good points}. Note that not all points are good, but such a better property is hereditary, in the sense that any sub-instance has good points. So, we can decompose all points into layers, by first finding out all good points in the original instance, removing them, and then iteratively finding good points in the residual instances.

These observations motivate the definition of $\tau$-SNC property proposed in \cite{Agarwal.}. Because of the above consideration of decomposition in the remaining instance, it is defined in a more general setting: restricted to any element set and any subcollection of sets.

\begin{definition}[$\tau$-collapsible and base group set]
{\rm For any subset $X'\subseteq X$ containing $e$ and $\mathcal S'\subseteq \mathcal N_{\mathcal S}(e)$, the {\em neighborhood of $e$ restricted to $(X', \mathcal{S}')$} is $RN(e,X',\mathcal S')=\mathcal U(\mathcal{N}_{\mathcal S'}(e))\cap X'$. We say that $RN(e,X',\mathcal S')$ is {\em $\tau$-collapsible} if there exists a collection $\mathcal B(e)$ consisting of at most $\tau$ sets from $\mathcal N_{\mathcal S'}(e)$ covering $RN(e,X',\mathcal S')$. Call $\mathcal B(e)$ as a {\em base group set of $e$ restricted to $(X', \mathcal{S}')$}.}
\end{definition}

\begin{example}\label{eg0112-2}
{\rm For the example in Fig. \ref{fig0112-1}, $X=\{p_1,\ldots,p_{10}\}$, $\mathcal S=\{S_1,\ldots,S_{16}\}$, where $S_i$ is the set of points covered by interval $I_i$. Suppose $\mathcal S'=\{S_2,S_3,S_8,S_{12},S_{14},S_{16}\}$ and $X'=\{p_2,p_3,p_5,p_7,p_8\}$. Then, $\mathcal N_{\mathcal S'}(p_7)=\{S_3,S_{12},S_{16}\}$, $RN(p_7,X',\mathcal S')=\{p_3,p_5,p_7,p_8\}$. Note that $RN(p_7,X',\mathcal S')$ is 2-collapsible since $\mathcal B(p_7)=\{S_{12},S_{16}\}$ is a base group set of $p_7$ restricted to $(X',\mathcal S')$.}
\end{example}

\begin{definition}[$\tau$-SNC]
{\rm Given a set system $(X,\mathcal{S})$, for a subset $X'\subseteq X$ and an element $e\in X'$, call $e$ a {\em $\tau$-SNC element in $X'$} if for any $\mathcal S'\subseteq \mathcal{N}_{\mathcal{S}}(e)$, $RN(e,X',\mathcal S')$ is $\tau$-collapsible. $(X,\mathcal{S})$ is said to have {\em $\tau$-SNC property} if for any $X'\subseteq X$, there exists an element $e\in X'$ which is a $\tau$-SNC element in $X'$.}
\end{definition}

The property of $\tau$-SNC is hereditary in the following sense, if $e$ is a $\tau$-SNC element in $X_1$, then for any $X_2\subseteq X_1$ with $e\in X_2$, $e$ is also a $\tau$-SNC element in $X_2$.

\begin{remark}
{\rm Interval cover problem has 1-SNC property. In fact, for any $X'\subseteq X$, the leftmost point $p$ of $X'$ is a 1-SNC element: for any $\mathcal S'\subseteq\mathcal N_{\mathcal S'}(p)$, the interval of $\mathcal S'$ with the rightmost endpoint covers $RN(p,X',\mathcal S')$. }
\end{remark}

For the instance in Fig. \ref{fig0112-1}, if we consider element set $X''=\{p_3,p_5,p_7,p_8\}$, then $p_3$ is a 1-SNC element. For example, for subcollection $\mathcal S''=\{S_1,S_3,S_7,S_{12},S_{16}\}$, interval $I_{12}\in\mathcal N_{\mathcal S''}(p_3)$ covers all points of $RN(p_2,X'',\mathcal S'')=\{p_3,p_5,p_7\}$, and thus a base group set of $p_3$ consists of only one set. As we have shown in Example \ref{eg0112-2} that $RN(p_7,X',\mathcal S')$ is only 2-collapsible. However, $X'$ does have 1-SNC elements, namely $p_2,p_5,p_8$. It should be noticed that empty set is a base group set of $p_2$ restricted to $(X',\mathcal S')$. The size of a base group set is only required to be {\em at most} $\tau$, not exactly $\tau$. Also notice that removing 1-SNC elements of $X'$, then $p_3,p_7$ become 1-SNC elements of the residual set $X'\setminus \{p_2,p_3,p_5\}$.

\begin{definition}[layer decomposition]\label{def1223-1}
{\rm Given a set system $\langle X,\mathcal{S} \rangle$ with SNC property, the {\em layer decomposition} of $(X,\mathcal{S})$ is a decomposition of $X$ into $X=\bigcup_{j=1}^L Z_j$, where $Z_1$ is the set of all $\tau$-SNC elements in $X$, and for $k=2,\dots,L$, $Z_k$ is  the set of all $\tau$-SNC elements in $X\setminus (\bigcup_{j=1}^{k-1} Z_j)$. We call $L$ as the {\em layer depth} of $(X,\mathcal{S})$, and the elements in $Z_j$ are said to have {\em layer level} $j$.}
\end{definition}
The following result was proved in \cite{Agarwal.}.
\begin{lemma}\label{lem1}
Given a set system $(X, \mathcal{S})$ and a constant $\tau$, there exists a procedure which can test whether $(X, \mathcal{S})$ has $\tau$-SNC property, and if it does, output the layer decomposition of $(X, \mathcal{S})$. This procedure takes $L$ iterations and can be implemented in parallel on $n^{O(\tau)}$ machines.
\end{lemma}
For any sub-collection $\mathcal S'\subseteq \mathcal S$ with the form $\mathcal S'=\{S_1,S_2,\ldots,S_{\ell}\}$, denote by $\mathcal S'_i=\{S_1,S_2,\ldots,S_{i}\}$ for $i=1,\ldots,\ell$. For any element $e$, denote by $L(e)$ the layer level of $e$.

\section{Parallel algorithm for MinWSC with $\tau$-SNC}\label{sec3}

In this section, we give a parallel algorithm for MinWSC with $\tau$-SNC property, using a primal-dual method.

\subsection{Algorithm}

MinWSC can be modeled as an integer program as follows, where $x(S)=1$ indicates that set $S$ is picked and $x(S)=0$ otherwise:

\begin{align}\label{1-1}
&\min\,\, \sum\limits_{S\in \mathcal S}x(S)\cdot w(S)\nonumber\\
&\mbox{s.t.}\quad
\begin{cases}
\sum\limits_{S\in \mathcal{S}\colon e\in S}x(S)\geq 1, &  \forall e\in X,\\
x(S)\in \{0,1\}, & \forall S\in \mathcal {S}.
\end{cases}
\end{align}
The integer program \eqref{1-1} can be relaxed to a linear program LP as follows:
\begin{align}\label{1-2}
&\min\,\, \sum\limits_{S\in \mathcal S}x(S)\cdot w(S)\nonumber\\
&\mbox{s.t.}\quad
\begin{cases}
\sum\limits_{S\in \mathcal{S}\colon e\in S}x(S)\geq 1, &  \forall e\in X,\\
x(S)\geq 0, & \forall S\in \mathcal {S}.
\end{cases}
\end{align}
Its dual program is as follows:
\begin{align}\label{1-3}
&\max\,\, \sum\limits_{e\in X}y(e)\nonumber\\
&\mbox{s.t.}\quad
\begin{cases}
\sum\limits_{e\in S}y(e)\leq w(S), &  \forall S\in \mathcal{S},\\
y(e)\geq 0, & \forall e\in X.
\end{cases}
\end{align}

The algorithm consists of two phases$\colon$ {\em forward phase} (line \ref{line4} to \ref{line002} of Algorithm \ref{algo1}) and {\em deletion phase} (Algorithm \ref{algo2}). The forward phase is to construct a collection of subcollections $\mathcal A$ which covers all elements. The deletion phase is to remove some redundant sets. In the deletion phase of \cite{Agarwal.}, the authors defined a {\em primal slackness property} with parameter $\gamma$ as follows: for any element $e$ with $y(e)>0$, $\sum\limits_{S\in \mathcal S\colon e\in S}x(S)\leq \gamma$. Their $\left(1+(2+\varepsilon)\tau^2)\right)$-approximation was derived by showing that $\gamma\leq \tau^2$. Our algorithm can reduce $\gamma$ from $\tau^2$ to down to $\tau$.

The forward phase essentially employs a primal-dual schema: starting from $\{y(e)\equiv 0\}_{e\in X}$, dual variables are increased until some constraint for the dual LP becomes {\em nearly tight}, where a dual constraint corresponding to set $S$ (see \eqref{1-3}) is nearly tight if the remaining weight $w'(S)=w(S)- \sum_{e\in S}y(e)$ satisfies $w'(S)\leq\varepsilon w(S)$, at which time, all nearly tight sets are picked (line \ref{line2} and line \ref{line1222-3} of Algorithm \ref{algo1}). Note that elements are dealt with layer by layer. Call each round of the for loop of Algorithm \ref{algo1} as an {\em epoch}. The $k$th {\em epoch} tries to cover all elements in $F^k$ using a collection $\mathcal C_k$ of picked subcollections, where $F^k$ denotes the set of elements of $Z_k$ which are not covered by subcollections picked in previous epochs. In order to efficiently control the number of rounds, the increase of dual variables is increasing in a geometric manner (see line \ref{line0-0}, \ref{line1222-1}, and \ref{eq1222-2} of Algorithm \ref{algo1}). After the forward phase, we get a feasible solution $\mathcal A$ for MinWSC.

\begin{algorithm}[H]
\caption{ParaMinSC-SNC$(\mathcal{G}=(X,\mathcal{S}, w))$}
\textbf{Input}$\colon$ A set system $(X,\mathcal{S})$ with $\tau$-SNC property and a weight function $w$ on $\mathcal{S}$

\textbf{Output}$\colon$ A feasible set cover $\mathcal{A}^\prime$ of $\mathcal{G}=(X,\mathcal{S}, w)$

\begin{algorithmic}[1]\label{algo1}

\STATE compute the layer decomposition $Z_1,\dots,Z_{L}$ in parallel \label{line4}

\STATE $\mathcal A\leftarrow\emptyset$; $X'\leftarrow  X$; $\mathcal{S}'\leftarrow  \mathcal{S}$;
\STATE $y(e)\leftarrow  0$ for each $e \in X$; $w'(S)\leftarrow  w(S)$ for each $S\in \mathcal{S}$ \label{linee0}
\STATE find a layer decomposition of $X$ using a parallel algorithm as in Lemma \ref{lem1}
\FOR {$k=1$ to $L$}\label{line1}
\STATE $\mathcal C_k=\emptyset$; $F^{k}\leftarrow X'\cap Z_k$; $t\leftarrow 1$\label{line1223-1}
\STATE $q_k\leftarrow  \min_{S\in \mathcal{S}'}w'(S)/|S\cap X'\cap Z_k|$; $\alpha_{k,t}\leftarrow  q_k$ \label{line0-0}
\WHILE{$\mathcal{A}$ does not cover $Z_k$}
        \STATE $y(e)\leftarrow y(e)+\alpha_{k,t} $ for each $e\in X'\cap Z_k$ \label{line1222-1}
         \STATE  $w'(S)\leftarrow w(S)- \sum_{e\in S}y(e) $ for each $S\in \mathcal S'$ \label{line1224}
    \STATE $\mathcal{C}'\leftarrow  \{S\in \mathcal{S}'\colon w'(S)\leq \varepsilon w(S)\}$\label{line2}
    \STATE $\mathcal{C}_{k}\leftarrow  \mathcal C_k\cup \mathcal C'$; $\mathcal A\leftarrow \mathcal A\cup \mathcal C'$\label{line1222-3}
    \STATE $X'\leftarrow  X'\setminus \mathcal{U}(\mathcal{C}')$; $\mathcal{S}'\leftarrow  \mathcal{S}'\setminus \mathcal C'$\label{line1230}
    \STATE $\alpha_{k,t}\leftarrow  q_k\varepsilon(1+\varepsilon)^{t-1}$; $t\leftarrow  t+1$\label{eq1222-2}
  \ENDWHILE
\ENDFOR\label{line002}
\STATE return $\mathcal A'\leftarrow$ Delete($\{\mathcal C_k\}_{k=1}^L,\{F^{k}\}_{k=1}^L$)\label{line1119}

\end{algorithmic}
\end{algorithm}

The goal of the deletion phase is to ensure that the sub-collection $\mathcal{A}^\prime$ output in line \ref{line1119} of Algorithm \ref{algo1} is a feasible solution of MinWSC satisfying the following property:
\begin{equation}\label{eq1125-1}
\mbox{for any $e\in E$ with $y(e)>0$, there are at most $\tau$ sets of $\mathcal A'$ covering $e$.}
\end{equation}
To realize this goal, elements are dealt with in reverse order from layer $L$ down to layer 1. When dealing with elements in layer $k$, $\mathcal B_k$ is kept to be a collection of sets covering
\begin{equation}\label{eq0115-1}
Q_k=\bigcup_{j=k}^L\left(\mathcal U(\mathcal C_j)\cap (\bigcup_{i=j}^LZ_i)\right).
\end{equation}
Initially, $\mathcal B_k=\mathcal A'\cup\mathcal C_k$. Some element $e\in Q_k$ with $y(e)>0$ might be covered by more than $\tau$ sets of $\mathcal B_k$. So, we have to shrink $\mathcal B_k$ to satisfy property \ref{eq1125-1}. Because every element $e\in Q_k$ is a $\tau$-SNC element of $\bigcup_{i=k}^LZ_i$ (by the definition of layer decomposition in Definition \ref{def1223-1}), the neighbors of $e$ in $\bigcup_{i=k}^LZ_i$ can be covered by a base group set consisting of at most $\tau$ sets. So, an idea is to select for each element a base group set from $\mathcal B_k$. However, there is a synchronous problem: it might happen that a set covering both $e$ and $e'$ is picked into the base group set of $e$, but is not picked into the base group set of $e'$, and thus it is still possible for $e'$ to be covered by more than $\tau$ sets of the union of these base group sets. To avoid such an asynchronous problem, the algorithm finds base group sets for a set of maximal independent elements, where a set $I$ is a {\em maximal independent set} if no elements in $I$ are neighbors of each other and adding any element into $I$ destroys the independence property. To realize this idea, a sequence of auxiliary graphs $G_k,\ldots,G_L$ and a sequence of maximal independent sets $I_k,\ldots,I_L$ are constructed as follows: the vertices of $G_j$ are those elements in $Z_j$ not adjacent with $\bigcup_{i=k}^{j-1}I_i$, \begin{equation}\label{eq1216}
\mbox{two vertices of $G_j$ are adjacent if they are neighboring elements in $\mathcal B_k$,}
\end{equation}
and $I_j$ is a maximal independent set of $G_j$. From such a construction,
\begin{equation}\label{eq0114-2}
I=\bigcup_{j=k}^L I_j\ \mbox{is an independent set in $\mathcal B_k$.}
\end{equation}
For each element $e\in I$, find a base group set $\mathcal B(e)$. In line \ref{line1212} of Algorithm \ref{algo2}, $\mathcal B_k$ is set to be the union of these base group sets. Note that
\begin{equation}\label{eq0114-3}
\{\mathcal N_{\mathcal B_k}(e)\}_{e\in I} \ \mbox{are disjoint collection of sets.}
\end{equation}
In fact, if $S$ is a common set of $\mathcal N_{\mathcal B_k}(e)$ and $\mathcal N_{\mathcal B_k}(e')$ for $e\neq e'$, then $e,e'$ are neighbors in $\mathcal B_k$, contradicting \eqref{eq0114-2}. We shall prove in Claim 1 of Lemma \ref{lem11} that any element in $F^{k}$ are covered by at most $\tau^2$ sets of $\mathcal B_k$. This part is essentially the same as that in \cite{Agarwal.}. To further shrink $\mathcal B_k$ to satisfy property \eqref{eq1125-1}, for each element $e\in I$, we flip a fair coin. If it is head, then a subcollection $\mathcal Q(e)$ of $\mathcal N_{\mathcal B_k}(e)$ consisting of $|\mathcal N_{\mathcal B_k}(e)|-1$ sets is picked in a random manner. If it is tail, then all sets in $\mathcal N_{\mathcal B_k}(e)$ are picked into $\mathcal Q(e)$. By \eqref{eq0114-3},
\begin{equation}\label{eq0114-5}
\mbox{how $\{\mathcal Q(e)\}_{e\in I}$ are taken are independent events.}
\end{equation}
The role that $\{\mathcal Q(e)\}_{e\in I}$ plays is to help with finding a collection of sets $\mathcal R'$ which can be deleted from $\mathcal B_k$ without affecting the covering requirement (see line \ref{line1216-3} to \ref{line1205-3}). We shall prove in Lemma \ref{lem11} that after $4\tau^{3}2^\tau \log n$ rounds of the while loop, the elements in $F^k$ satisfy property \eqref{eq1125-1} with high probability.

\begin{algorithm}[H]
\caption{Delete$(\{\mathcal C_k\}_{k=1}^L,\{F^{k}\}_{k=1}^L)$}
\textbf{Input}$\colon$ $\{\mathcal C_k\}_{k=1}^L,\{F^{k}\}_{k=1}^L$ computed by Algorithm \ref{algo1}

\textbf{Output}$\colon$ A feasible solution  $\mathcal{A}^\prime$ of $\mathcal{G}=(X,\mathcal{S}, w)$ satisfying property \eqref{eq1125-1}.

\begin{algorithmic}[1]\label{algo2}
\STATE $\mathcal A'\leftarrow \emptyset$
\FOR {$k=L$ to 1}\label{line1202}
\STATE $i\leftarrow 0$; $I\leftarrow \emptyset$
\STATE $\mathcal B_k\leftarrow \mathcal A'\cup \mathcal C_k$\label{line1211}
\FOR {$j=k$ to $L$}\label{line1213}
\STATE find a maximal independent set $I_j$ of $G_j$ constructed in \eqref{eq1216}
\STATE $I\leftarrow I\cup I_j$
\ENDFOR \label{label0102-1}
\STATE for each $e\in I$, find a base group set $\mathcal B(e)$ for $e$ restricted to $(\bigcup_{j=k}^LZ_j,\mathcal N_{\mathcal B_k}(e))$ \label{line1213-1}
\STATE $\mathcal B_k\leftarrow \bigcup_{e\in I}\mathcal B(e)$\label{line1212}
\WHILE {$i\leq 4\tau^{3}2^\tau \log n$}
\FOR {each $e\in I$ in parallel}
\STATE flip a fair coin
\IF {head}
\STATE $\mathcal Q(e)\leftarrow$ $(|\mathcal N_{\mathcal B_k}(e)|-1)$ sets picked randomly and uniformly from $\mathcal N_{\mathcal B_k}(e)$\label{line1227}
\ELSE
\STATE $\mathcal Q(e)\leftarrow \mathcal N_{\mathcal B_k}(e)$
\ENDIF
\ENDFOR
\STATE $\mathcal R'\leftarrow \emptyset$\label{line1216-3}
\FOR {each $S\in \mathcal B_k\setminus \left(\bigcup_{e\in I}\mathcal Q(e)\right)$ in parallel}\label{line1208}
\IF {$\mathcal U(S)\cap (\bigcup_{j=k}^LZ_j)\subseteq \mathcal U(\bigcup_{e\in I}\mathcal Q(e))$}\label{label1229-1}
\STATE $\mathcal R'\leftarrow \mathcal R'\cup \{S\}$\label{line1209}
\ENDIF
\ENDFOR\label{line1208-1}
\STATE $\mathcal B_k\leftarrow \mathcal B_k\setminus\mathcal R'$\label{line1205-3}
\STATE $i\leftarrow i+1$
\ENDWHILE
\STATE $\mathcal A'\leftarrow \mathcal B_k$\label{line1220}
\ENDFOR
\end{algorithmic}
\end{algorithm}

%
%

\subsection{Analysis}

The following lemma shows the feasibility of $\mathcal A'$ output by Algorithm \ref{algo1}.
\begin{lemma}\label{lem1105}
The collection $\mathcal A'$ returned by Algorithm \ref{algo1} is a set cover.
\end{lemma}
\begin{proof}
We shall prove by induction on $k$ from $L$ down to $1$ that
\begin{equation}\label{eq1220-1}
\mbox{in Algorithm \ref{algo2}, after processing layer $k$, $\mathcal B_k$ covers all elements of $Q_k$,}
\end{equation}
where $Q_k$ is defined in \eqref{eq0115-1}. Then the lemma follows from \eqref{eq1220-1} for $k=1$ since $X=Q_1$.

First consider the inductive basis when $k=L$. Initially $\mathcal B_L=\mathcal C_L$, which covers $Q_L=\mathcal U(\mathcal C_L)\cap Z_L$. For any $e\in Q_L$, if $e\in I_L$, then $e$ is covered by $\mathcal B_L$ constructed in line \ref{line1212}. Otherwise $e\notin I_L$, and by the maximality of $I_L$, element $e$ is adjacent with an element $e'\in I_L$. By the SNC property, base group set $\mathcal B(e')$ covers $e$. In any case, $\mathcal B_L$ in line \ref{line1212} covers $Q_L$. Since line \ref{line1208} to line \ref{line1205-3} only removes redundant sets, $\mathcal B_L$ in line \ref{line1205-3} still covers all elements of $Q_L$. The inductive basis is proved.

Suppose \eqref{eq1220-1} holds for $k+1$ with $k\leq L-1$. By the induction hypothesis, $\mathcal A'=\mathcal B_{k+1}$ covers $Q_{k+1}$. Adding $\mathcal C_k$ into $\mathcal A'$, the resulting $\mathcal B_k$ in line \ref{line1211} covers $Q_k\subseteq Q_{k+1}\cup\mathcal U(\mathcal C_k)$. By a similar argument as the above, after removing redundant sets, $\mathcal B_k$ covers $Q_k$. The induction step is completed.
\end{proof}

The next lemma gives some important properties of $\mathcal A'$.

\begin{lemma}\label{lem11}
The collection $\mathcal{A}^\prime$ returned by Algorithm \ref{algo1} satisfies the following properties:

$(\romannumeral1)$ with high probability, any $e\in \bigcup_{j=1}^LF^{j}$ is covered by at most $\tau$ sets of $\mathcal{A}^\prime$;

$(\romannumeral2)$ for any $S\in \mathcal{A}^\prime$, $(1-\varepsilon)w(S)\leq\sum_{e\in S}\,y(e)\leq w(S)$, where $\{y(e)\}_{e\in X}$ is the set of dual variables at the end of the algorithm.
\end{lemma}

\begin{proof}[Proof of property $(\romannumeral1)$] It suffices to prove that after the $k$th epoch of Algorithm \ref{algo2}, which is the $k$th round of the outer for loop,
\begin{equation}\label{eq1220-2}
\mbox{any element in $F^k$ is covered by at most $\tau$ sets of $\mathcal A'$ with high probability.}
\end{equation}
In fact, notice that for any $j<k$, no set of $\mathcal C_j$ can cover any element in $F^k$ (because of the definition of $F^k$ in line \ref{line1223-1} of the forward phase). So, as long as property \eqref{eq1220-2} can be proved after the $k$th epoch, this property is maintained throughout processing layers $k-1,k-2,\ldots,1$, and thus property $(\romannumeral1)$ follows.

In the following, all labels of lines we mention refer to Algorithm \ref{algo2}. We first bound $|\mathcal N_{\mathcal B_k}(e)|$ for any element $e\in F^{k}$, where $\mathcal B_k$ is the collection in line \ref{line1212}.

\vskip 0.2cm {\bf Claim 1.} For the collection $\mathcal B_k$ in line \ref{line1212}, any element $e\in F^{k}$ satisfies $|\mathcal N_{\mathcal B_k}(e)|\leq \tau^2$.

If $e\in I$, then those sets in $\mathcal B(e)$ are the only sets of $\mathcal B_k$ covering $e$. In fact, if there is a set $S\in\mathcal B(e')$ covering $e$, where $e'$ is an element of $I$ which is different from $e$, then $S$ covers both $e$ and $e'$, and thus $e$ and $e'$ are neighbors in $\mathcal B_k$, contradicting that $I$ is an independent set of $\mathcal B_k$. It follows that
\begin{equation}\label{eq1228-1}
|\mathcal N_{\mathcal B_k}(e)|=|\mathcal B(e)|\leq \tau\ \mbox{for any element $e\in I$.}
\end{equation}

Next, consider an element $e\in F^k\setminus I$. We first prove that
\begin{equation}\label{eq1223-3}
I\ \mbox{has at most $\tau$ elements which are neighbors of $e$ in $\mathcal B_k$}.
\end{equation}
Denote by $E'$ the set of elements in $I$ which are neighbors of $e$ in $\mathcal B_k$, and let $\mathcal B'=\bigcup_{e'\in E'}\{$an arbitrary set in $\mathcal R_{\mathcal B_k}(e',e)\}$ (recall that $\mathcal R_{\mathcal B_k}(e',e)$ is the collection of sets in $\mathcal B_k$ containing both $e$ and $e'$). Since $e$ is a $\tau$-SNC element of $\bigcup_{j=k}^LZ_j$, every set of $\mathcal B'$ contains $e$, and elements in $E'$ are neighbors of $e$ in $\mathcal B'$, there exists a subcollection $\mathcal B''\subseteq \mathcal B'$ with size at most $\tau$ such that $\mathcal B''$ covers $E'\cup \{e\}$. If $|E'|\geq \tau+1$, then by the pigeonhole principle, there exists a set $S\in \mathcal B''$ covering two elements of $E'$, contradicting that $E'$ is independent. So, \eqref{eq1223-3} is proved.

Notice that a set $S\in \mathcal N_{\mathcal B_k}(e)=\mathcal N_{\bigcup_{e'\in I}\mathcal B(e')}(e)$ implies that $e\in S$ and there exists an element $e'\in I$ such that $S\in\mathcal B(e')$. So, $S$ contains both $e$ and $e'$, and thus $e'$ is a neighbor of $e$ in $I$. Combining this with \eqref{eq1223-3} and the fact that every base group set has size at most $\tau$, we have
$$
|\mathcal N_{\mathcal B_k}(e)|=|\mathcal N_{\bigcup_{e'\in I}\mathcal B(e')}(e)|\leq\tau^2.
$$
Claim 1 is proved.

\vskip 0.2cm {\bf Claim 2.} Let $E''=\{e\in F^k\colon |\mathcal N_{\mathcal B_k}(e)>\tau\}$, where $\mathcal B_k$ is the collection of sets in line \ref{line1212}. At the end of the while loop of Algorithm \ref{algo2}, the updated $\mathcal B_k$ satisfies that for any element $e\in E''$, $|\mathcal N_{\mathcal B_k}(e)|\leq \tau$ with high probability.

By \eqref{eq1228-1}, $E''\cap I=\emptyset$. For any $e\in E''$, let $A(e)$ be the event that $|\mathcal N_{\mathcal B_k}(e)|$ is decreased by at least 1 after one while loop. We shall prove that probability $Pr(A(e))$ satisfies
\begin{equation}\label{eq1221}
Pr(A(e))\geq \frac{1}{2^\tau \tau}.
\end{equation}

To prove \eqref{eq1221}, consider collection $\mathcal B_k$ in line \ref{line1205-3} at the end of the $i$-th while loop, and a base group set $\mathcal{B}(e)$ restricted to $(\mathcal U(\mathcal N_{\mathcal B_k}(e))\cap \bigcup_{j=k}^LZ_j,\mathcal N_{\mathcal B_k}(e))$. The idea is: in the $(i+1)$-th while loop, for the selected collection of sets $\{\mathcal Q(e')\}_{e'\in I}$, if $\mathcal{B}(e)\subseteq \bigcup_{e'\in I}\mathcal Q(e')$, then by the SNC property, any set $S'\in \mathcal N_{\mathcal B_k}(e)\setminus \bigcup_{e'\in I}\mathcal Q(e')$ satisfies $\mathcal U(S')\cap \bigcup_{j=k}^LZ_j\subseteq \mathcal U(\mathcal{B}(e))\subseteq \mathcal U(\bigcup_{e'\in I}\mathcal Q(e'))$, meeting the condition in line \ref{label1229-1}, and thus can be deleted in line \ref{line1205-3}, resulting in event $A(e)$ if $\mathcal N_{\mathcal B_k}(e)\setminus \bigcup_{e'\in I}\mathcal Q(e')\neq\emptyset$. So in the following, we lower bound the probability of $\big[\mathcal{B}(e)\subseteq \bigcup_{e'\in I}\mathcal Q(e')\big]\wedge \big[\mathcal N_{\mathcal B_k}(e)\setminus \bigcup_{e'\in I}\mathcal Q(e')\neq\emptyset\big]$.

For each $S\in\mathcal N_{\mathcal B_k}(e)$, there exists an element $e'\in I$ with $S\in \mathcal N_{\mathcal B_k}(e')$, since every set in $\mathcal B_k$ comes from a base group set of some $e'\in I$. Let $I'=\{e'\in I\colon \mathcal{B}(e)\cap \mathcal N_{\mathcal B_k}(e')\neq\emptyset\}$. Note that
\begin{equation}\label{eq0115-5}
\mbox{any $e'\in I'$ is a neighbor of $e$.}
\end{equation}
For $e'\in I'$, let $A'(e')$ be the event that $\mathcal{B}(e)\cap \mathcal N_{\mathcal B_k}(e')$ remains to be in $\mathcal Q(e')$. Its probability
\begin{equation}\label{eq1230-1}
Pr(A'(e'))\geq Pr(\mathcal Q(e')\leftarrow \mathcal N_{\mathcal B_k}(e'))=1/2.
\end{equation}

Since $|\mathcal N_{\mathcal B_k}(e)|>\tau$ and $|\mathcal{B}(e)|\leq \tau$, we have $\mathcal N_{\mathcal B_k}(e)\setminus \mathcal{B}(e)\neq\emptyset$. Consider a set $S'\in \mathcal N_{\mathcal B_k}(e)\setminus \mathcal{B}(e)$, let $e''$ be an element in $I$ with $S'\in \mathcal N_{\mathcal B_k}(e'')$. For simplicity of statement, we only argue about the case when $e''\in I'$ (the case $e''\in I\setminus I'$ can be obtained similarly with an even simpler argument). In the case $e''\in I'$, we have to consider the probability $Pr(A'(e'')\wedge S'\not\in \mathcal Q(e''))$, which equals
$$
Pr(|\mathcal Q(e'')|=|\mathcal N_{\mathcal B_k}(e'')|-1)\cdot Pr(\mbox{$S'$ is the set not picked into $\mathcal Q(e'')$}).
$$
Since $1\leq |\mathcal N_{\mathcal B_k}(e'')|\leq \tau$, we have
$$
Pr(\mbox{$S'$ is the set not picked into $\mathcal Q(e'')$})=\frac{1}{|\mathcal N_{\mathcal B_k}(e'')|}\geq\frac{1}{\tau}.
$$
Hence
\begin{equation}\label{eq1231-1}
Pr(A'(e'')\wedge S'\not\in \mathcal Q(e''))\geq\frac{1}{2\tau}.
\end{equation}

Since any element in $I'$ (including $e''$) is a neighbor of $e$ (see \eqref{eq0115-5}), by \eqref{eq1223-3}, we have \begin{equation}\label{eq1231-3}
|I'\setminus\{e''\}|\leq \tau-1.
\end{equation}
Since we have shown in \eqref{eq0114-5} that how $\{\mathcal Q(e')\}_{e'\in I}$ are taken are independent events, the events $\{A'(e')\}_{e'\in I'\setminus \{e''\}}$ as well as the event $A'(e'')\wedge \big[S'\not\in \mathcal Q(e'')\big]$ are mutually independent. Combining this observation with \eqref{eq1230-1}, \eqref{eq1231-1}, \eqref{eq1231-3}, the probability
$$
P(\big[\mathcal{B}(e)\subseteq \bigcup_{e'\in I}\mathcal Q(e')\big]\wedge \big[\mathcal N_{\mathcal B_k}(e)\setminus \bigcup_{e'\in I}\mathcal Q(e')\neq\emptyset\big])\geq \left(\frac12\right)^{|I'\setminus\{e''\}|}\cdot \frac{1}{2\tau}\geq \left(\frac12\right)^{\tau-1}\cdot \frac{1}{2\tau}=\frac{1}{2^\tau \tau},
$$
and inequality \eqref{eq1221} follows.

Denote by $T(e)$ the number of rounds for $|\mathcal N_{\mathcal B_k}(e)|$ to decrease below $\tau$. Combining Claim 1 and inequality \eqref{eq1221}, we have
$$
E(T(e))\leq \frac{\tau^2-\tau}{Pr(A(e))}\leq \tau^{3}2^\tau.
$$
By Markov's inequality,
\begin{equation}\label{eq1205-1}
Pr\left(T(e)>2\tau^{3}2^\tau \right)\leq 1/2.
\end{equation}

Let $C(e)$ be the event that $|\mathcal N_{\mathcal B_k}(e)|\leq \tau$ after $4\tau^{3}2^\tau \log n$ rounds. By inequality \eqref{eq1205-1},
\begin{equation}
Pr(\overline{C(e)})\leq (1/2)^{\frac{4\tau^{3}2^\tau\log n}{2\tau^{3}2^\tau }}=\frac{1}{n^2}.
\end{equation}
By the union bound,
$$
Pr(\bigcup_{e\in F^k}\overline{C(e)})\leq \frac{1}{n}, \ \mbox{and thus}\ P(\bigwedge_{e\in F^k}C(e))\geq 1-\frac{1}{n}.
$$
Claim 2 is proved. And then, by the argument at the beginning, property $(\romannumeral1)$ is proved.
\end{proof}

\begin{proof}[Proof of property $(\romannumeral2)$] Assume that the number of iterations in the $k$th epoch of Algorithm \ref{algo1} is $T_k$. For the $t$th iteration of the while loop in the $k$th epoch, denote by $w_{k,t}(S)$ the residual weight $w'(S)$ in line \ref{line1224}, $X_{k,t}$ the residual element set $X'$ in line \ref{line1230}, and $y_{k,t}(e)$ the current dual variable for element $e$. Suppose $S$ is picked into $\mathcal A$ in the $t(S)$th iteration of the $k(S)$th epoch. By the criteria of how $S$ can be picked (see line \ref{line2}), we have
\begin{align}\label{eq1-1}
w(S)- \sum_{e\in S}y_{k(S),t(S)}(e)=w_{k(S),t(S)}(S)\leq\varepsilon w(S).
\end{align}
Note that after $S$ is picked, any element $e\in S$ has its dual variable $y(e)$ frozen (which means that $y(e)$ does not increase its value any more), and thus \begin{equation}\label{eq0101-1}
y(e)=y_{k(S),t(S)}(e)\ \mbox{for any $e\in S$.}
\end{equation}
Hence the final dual variables $\{y(e)\}_{e\in X}$ satisfy
$$
\sum_{e\in S}y(e)=\sum_{e\in S}y_{k(S),t(S)}(e)\geq(1-\varepsilon)w(S).
$$
The left side of $(\romannumeral2)$ is proved.

To prove the right side of $(\romannumeral2)$, notice that $w_{k(S),t(S)}(S)=w(S)- \sum_{e\in S}y(e)$ holds from the time that $S$ is picked, so it suffices to show that
\begin{align}\label{eq0001-2}
w_{k(S),t(S)}(S)\geq 0.
\end{align}
Since $S$ is not picked before the $k(S)$th epoch, we have $w'(S)>\varepsilon w(S)>0$ at the beginning of the $k(S)$th epoch. Suppose $w_{k(S),t(S)}(S)<0$, then $t(S)>1$, $w_{k(S),t(S)-1}(S)> \varepsilon w(S)$, and in the $(t(S)-1)$th iteration, every element $e\in X_{k(S),t(S)-1}\cap Z_{k(S)}$ has its $y$-value increased by $\alpha_{k(S),t(S)-1}$. It follows that
\begin{align}\label{eq1-3}
w_{k(S),t(S)}(S)=w_{k(S),t(S)-1}(S)-\alpha_{k(S),t(S)-1}|S\cap X_{k(S),t(S)-1}\cap Z_{k(S)}|.
\end{align}
Under the assumption $w_{k(S),t(S)}(S)<0$, we have
\begin{align}\label{eq0001-4}
\alpha_{k(S),t(S)-1}|S\cap X_{k(S),t(S)-1}\cap Z_{k(S)}|>w_{k(S),t(S)-1}(S)>\varepsilon w(S).
\end{align}
It can be calculated that
\begin{align*}
w(S) & =w_{k(S),t(S)-1}+\sum\limits_{k=1}^{k(S)-1}\sum\limits_{t=1}^{T_k}\alpha_{k,t}|S\cap X_{k,t}\cap Z_k|+\sum\limits_{t=1}^{t(S)-2}\alpha_{k(S),t}|S\cap X_{k(S),t}\cap Z_{k(S)}|\nonumber \\
& \geq \sum\limits_{t=1}^{t(S)-2}\alpha_{k(S),t}|S\cap X_{k(S),t}\cap Z_{k(S)}|\\
&\geq\sum_{t=1}^{t(S)-2}\alpha_{k(S),t}|S\cap X_{k(S),t(S)-1}\cap Z_{k(S)}| \nonumber \\
&=|S\cap X_{k(S),t(S)-1}\cap Z_{k(S)}|\cdot q_{k(S)} (1+\sum_{t=0}^{t(S)-3}\varepsilon(1+\varepsilon)^{t})\nonumber \\
&=|S\cap X_{k(S),t(S)-1}\cap Z_{k(S)}|\cdot q_{k(S)}(1+\varepsilon)^{t(S)-2}\nonumber \\
&=|S\cap X_{k(S),t(S)-1}\cap Z_{k(S)}|\cdot \alpha_{k(S),t(S)-1}\cdot\frac{1}{\varepsilon} \nonumber \\
&> w(S),
\end{align*}
where the last inequality comes from \eqref{eq0001-4}. This contradiction establishes \eqref{eq0001-2}, and thus property $(\romannumeral2)$ is proved.
\end{proof}

\begin{lemma}\label{lem2}
For any $k\in\{1,\dots, L\}$, the number of rounds during the $k$th epoch in the forward phase is at most $O(\log_{1+\varepsilon}\frac{nw_{\max}}{\varepsilon w_{\min}})$, where $w_{\max}=\max\{ w(S)\colon S\in \mathcal{S}\}$ and $w_{\min}=\min\{ w(S)\colon S\in \mathcal{S}\}$.
\end{lemma}

\begin{proof}
Suppose the $k$th epoch is executed $T_k$ rounds. Consider the set $S$ which is picked into $\mathcal A$ in the last round of the $k$th epoch. Notice that for any $t=1,\ldots,T_k$, we have $\mathcal U(S)\cap X_{k,t}\cap Z_k\neq\emptyset$, and thus the reduction of the weight from $S$ in the $k$th epoch is
\begin{align}\label{eq1-6}
\sum\limits_{t=1}^{T_k}\alpha_{k,t}|\mathcal U(S)\cap X_{k,t}\cap Z_k| \geq \sum_{t=1}^{T_k}\alpha_{k,t}=q_k+\sum_{t=1}^{T_k-1} q_k\varepsilon(1+\varepsilon)^{t-1}=q_k(1+\varepsilon)^{T_k-1}.
\end{align}
By \eqref{eq0001-2}, this amount of reduction is upper bounded by $w(S)\leq w_{\max}$. By the definition of $q_k$ in line \ref{line0-0} of Algorithm \ref{algo1}, we have $q_k>\frac{\varepsilon w_{\min}}{n}$ (note that any set $S'$ which is not picked before the $k$th epoch has residual weight $w'(S')>\varepsilon w(S')\geq \varepsilon w_{\min}$ and covers at most $n$ elements). Combining these bounds with \eqref{eq1-6}, we have
$T_k\leq \log_{1+\varepsilon}\frac{nw_{\max}}{\varepsilon w_{\min}}+1$.
\end{proof}

To bound $\frac{w_{\max}}{w_{\min}}$, we modify Algorithm \ref{algo1} as follows.

\begin{algorithm}[H]
\caption{Modified-ParaMinSC-SNC$(\mathcal{G}=(X,\mathcal{S}, w))$}
\textbf{Input}$\colon$ A set system $(X,\mathcal{S},w)$ with $\tau$-SNC property.

\textbf{Output} A collection of sets $\mathcal B$ which covers all elements of $X$

\begin{algorithmic}[1]\label{algo3}
\STATE $\beta\longleftarrow\max_{e\in X}\min_{S\colon S\in \mathcal{N}_{\mathcal{S}}(e)}w(S)$.
\STATE $\mathcal S^{(1)}\leftarrow   \{S\in \mathcal S\colon w(S)\in  [0,\frac{\beta\varepsilon}{n})\}$; $\mathcal S^{(2)}\leftarrow \{S\in \mathcal S\colon w(S)\in [\frac{\beta\varepsilon}{n},n\beta]\}$\label{line0001}
\STATE $X^{(1)}\leftarrow   \mathcal{U}(\mathcal S^{(1)})$; $X^{(2)}\leftarrow X\backslash X^{(1)}$
\FOR{each $e\in X^{(1)}$ in parallel} \label{linee6}
\STATE Select any set $S\in \mathcal N_{\mathcal S^{(1)}}(e)$ and add it into $\mathcal{B}_1$\label{line1117}
\ENDFOR
\STATE $\mathcal{B}_2\leftarrow \mbox{ ParaMinSC-SNC}(\mathcal{G}^\prime=(X^{(2)},\mathcal S^{(2)}, w))$ \label{linee8}
\RETURN $\mathcal B\leftarrow \mathcal{B}_1\cup \mathcal{B}_2$
\end{algorithmic}
\end{algorithm}

\begin{theorem}\label{thm2}
With high probability, Algorithm \ref{algo3} has approximation ratio at most $(1+3\varepsilon)\tau$ in $O(L\log_{1+\varepsilon}\frac{n^3}{\varepsilon^2}+ 4\tau^{3}2^\tau L^2\log n)$ rounds on $O(nm^{\tau})$ machines, where $0< \varepsilon<\frac{1}{2}$.
\end{theorem}
\begin{proof}
Notice that if we take for each element $e$ the set in $\mathcal N_{\mathcal S}(e)$ with the minimum weight, then the union of these sets form a set cover, whose total weight is no more than $n\beta$. So, no set with weight larger than $n\beta$ can be used in an optimal solution. This is why it is sufficient for the algorithm to consider those sets with weight $\leq n\beta$.

Let $\mathcal B^*$ be an optimal solution of $\mathcal{G}=(X,\mathcal{S}, w)$, and $\mathcal B^*_2$ an optimal solution of $\mathcal{G}^\prime=(X^{(2)},\mathcal S^{(2)}, w)$. Notice that $\mathcal B^*\cap \mathcal S^{(2)}$ is a feasible solution to $\mathcal{G}^\prime$, so
\begin{align}\label{eq0003-01}
w(\mathcal B^*_2)\leq w( \mathcal B^*\cap \mathcal S^{(2)})\leq w( \mathcal B^*).
\end{align}
By line \ref{line0001} and line \ref{line1117} of Algorithm \ref{algo3},
\begin{align}\label{eq0003-02}
w(\mathcal{B}_1)\leq |X^{(1)}|\frac{\beta\varepsilon}{n}\leq \beta\varepsilon \leq w(\mathcal B^*)\varepsilon.
\end{align}
By Lemma \ref{lem11}, it can be estimated that
\begin{align}\label{eq1117}
w(\mathcal{B}_2) & =\sum\limits_{S\in \mathcal B_2}w(S)\leq \frac{\sum\limits_{S\in \mathcal B_2}\sum\limits_{e\in S}y(e)}{1-\varepsilon}=\frac{\sum\limits_{e\in S}y(e)\cdot |\{S\in \mathcal B_2\colon e\in S\}|}{1-\varepsilon} \nonumber\\
& \leq \frac{\tau}{1-\varepsilon}\sum\limits_{e\in E}y(e)\leq\frac{\tau}{1-\varepsilon} w(\mathcal B^*_2)\leq \frac{\tau}{1-\varepsilon}w(\mathcal B^*),
\end{align}
where the first inequality uses property $(\romannumeral2)$, the second inequality uses property $(\romannumeral1)$, the third inequality uses the weak duality theory (that is, the objective value of the dual variables $\{y(e)\}_{e\in X^{(2)}}$ is a lower bound for the objective value of the primal solution), and the last inequality uses \eqref{eq0003-01}.

Combining inequalities \eqref{eq0003-02} and \eqref{eq1117}, for $0<\varepsilon< 1/2$, we have
$$
w(\mathcal B)\leq w(\mathcal B^*)\varepsilon+ \frac{\tau}{1-\varepsilon} w(\mathcal B^*)=\frac{\varepsilon+\tau-\varepsilon^2}{1-\varepsilon}\cdot w(\mathcal B^*) \leq \tau(1+3\varepsilon)w(\mathcal B^*).
$$
the desired approximation ratio is obtained.

Next, we consider the number of rounds. For the reduced instance $\mathcal G'$, its heaviest set has weight at most $n\beta$, its lightest set has weight at least $\beta\varepsilon/n$. So, by Lemma \ref{lem2}, the forward phase for the reduced instance $\mathcal G'$ runs in $O(L\log_{1+\varepsilon}\frac{n^3}{\varepsilon^2})$ rounds on $(m+n)$ machines, where $m$ is the number of sets in $\mathcal S$. In the deletion phase, if we use the parallel algorithm in \cite{Ghaffari}, then finding a maximal independent set needs $O(\log n)$ rounds on $O(n)$ machines. Hence in the $k$th epoch of Algorithm \ref{algo2}, the steps from line \ref{line1213} to line \ref{label0102-1} run in $O((L-k+1)\log n)$ rounds on $O(n)$ machines; line \ref{line1213-1} needs $O(1)$ rounds on $O(nm^{\tau})$ machines \cite{Agarwal.}; and the while loop runs in $4\tau^{3}2^\tau\log n$ rounds on $(m+n)$ machines. Thus the total number of rounds for the deletion phase is $4\tau^{3}2^\tau L^2\log n$ on $O(nm^{\tau})$ machines. Adding together, calling function ParaMinSC-SNC$(\mathcal G')$ in line \ref{linee8} of Algorithm \ref{algo3} runs in $O(L\log_{1+\varepsilon}\frac{n^3}{\varepsilon^2}+ 4\tau^{3}2^\tau L^2\log n)$ rounds on $O(nm^{\tau})$ machines. All the other steps in Algorithm \ref{algo3} can be done parallelly in $O(1)$ rounds on $n$ machines. The adaptive complexity of the algorithm follows.
\end{proof}

\section{Conclusion and Discussion}\label{sec5}

This paper gave a parallel algorithm for the minimum set cover problem with $\tau$-SNC property, achieving approximation ratio $\tau(1+3\varepsilon)$ in $O(L\log_{1+\varepsilon}\frac{n^3}{\varepsilon^2}+4\tau^{3}2^\tau L^2\log n)$ rounds on $O(nm^{\tau})$ machines.

Our result affirmatively answers two of the three open questions proposed in \cite{Agarwal.}. Question $Q_{\uppercase\expandafter{\romannumeral3}}$ still remains, and new ideas are needed in order to answer it.

\section*{Acknowledgements}
This research is supported by NSFC (11901533, U20A2068, 11771013), and ZJNSFC (LD19A010001).


\end{document}